\documentclass[12pt]{article}

\usepackage[colorlinks=true,linkcolor=blue,citecolor=magenta]{hyperref}
\usepackage[mathscr]{euscript}
\usepackage{amssymb,amsmath,amsthm,mathtools,accents}
\usepackage{cite}
\usepackage{enumitem}

\usepackage[T1]{fontenc}
\usepackage{lmodern}
\usepackage{microtype}

\usepackage[letterpaper,hmargin=3.7cm,vmargin=3.3cm]{geometry}
\usepackage{setspace}
\makeatletter
\renewcommand{\section}{\@startsection%
{section}%
{1}%
{0em}%
{1.7em}%
{1.2em}%
{\normalfont\large\centering\bfseries}}
\renewcommand{\@seccntformat}[1]%
{\csname the#1\endcsname.\hspace{0.5em}}
\makeatother

\numberwithin{equation}{section}

\newtheorem{theorem}{Theorem}[section]
\newtheorem*{thm}{Theorem}
\newtheorem{proposition}[theorem]{Proposition}
\newtheorem{lemma}[theorem]{Lemma}

\theoremstyle{definition}
\newtheorem{definition}[theorem]{Definition}
\newtheorem{hypothesis}[theorem]{Hypothesis}

\theoremstyle{remark}
\newtheorem{remark}{Remark}

\newcommand{\ie}{{\it i.\,e.}}
\newcommand{\viz}{{\it viz.}}
\newcommand{\cf}{{\it cf.}}

\newcommand{\abs}[1]{\left\lvert #1 \right\rvert}
\newcommand{\abss}[1]{\lvert #1 \rvert}
\newcommand{\norm}[1]{\left\lVert #1 \right\rVert}

\newcommand{\inner}[2]{\left\langle#1,#2\right\rangle}
\newcommand{\inners}[2]{\big\langle#1,#2\big\rangle}

\newcommand{\cB}{\mathcal{B}}

\newcommand{\cO}{\mathcal{O}}

\newcommand{\cR}{\mathcal{R}}
\newcommand{\cPW}{\mathcal{PW}}

\newcommand{\R}{{\mathbb R}}
\newcommand{\C}{{\mathbb C}}
\newcommand{\N}{{\mathbb N}}
\newcommand{\Z}{{\mathbb Z}}
\newcommand{\cc}[1]{\overline{#1}}

\newcommand{\defeq}{\mathrel{\mathop:}=} 

\def\cprime{$'$}

\DeclareMathOperator{\im}{Im}
\DeclareMathOperator{\dom}{dom}

\DeclareMathOperator{\spec}{spec}

\begin{document}

\begin{titlepage}
\title{Oversampling and undersampling in de Branges spaces arising from
  regular Schr\"odinger operators
\footnotetext{%
Mathematics Subject Classification(2010):
47B32,  
41A05,  
34L40,  
}
\footnotetext{%
Keywords:
de Branges spaces;
Schr\"odinger operators;
Sampling Theory.
}
\hspace{-8mm}
}
\author{
\textbf{Luis O. Silva}\thanks{%
Supported by UNAM-DGAPA-PAPIIT IN110818 and SEP-CONACYT CB-2015 254062
}
\\
\small Departamento de F\'{i}sica Matem\'{a}tica\\[-1.6mm]
\small Instituto de Investigaciones en Matem\'aticas Aplicadas y en Sistemas\\[-1.6mm]
\small Universidad Nacional Aut\'onoma de M\'exico\\[-1.6mm]
\small C.P. 04510, Ciudad de M\'exico\\[-1.6mm]
\small \texttt{silva@iimas.unam.mx}
\\[2mm]
\textbf{Julio H. Toloza}\thanks{Partially supported by CONICET
	(Argentina) through grant PIP 11220150100327CO}
\\
\small INMABB\\[-1.6mm]
\small Departamento de Matem\'atica\\[-1.6mm]
\small Universidad Nacional del Sur (UNS) - CONICET\\[-1.6mm]
\small Bah\'ia Blanca, Argentina\\[-1.6mm]
\small \texttt{julio.toloza@uns.edu.ar}
\\[2mm]
\textbf{Alfredo Uribe}
\\
\small Departamento de F\'{i}sica Matem\'{a}tica\\[-1.6mm]
\small Instituto de Investigaciones en Matem\'aticas Aplicadas y en Sistemas\\[-1.6mm]
\small Universidad Nacional Aut\'onoma de M\'exico\\[-1.6mm]
\small C.P. 04510, Ciudad de M\'exico\\[-1.6mm]
\small \texttt{alfredo.uribe.83@ciencias.unam.mx}
}
\date{}
\maketitle
\vspace{-4mm}
\begin{center}
\begin{minipage}{5in}
  \centerline{{\bf Abstract}} \bigskip The classical results on
  oversampling and undersampling (or aliasing) of functions in
  Paley-Wiener spaces are generalized to the case of functions in de Branges spaces arising
  from regular Schr\"odinger operators with a wide range of potentials.
\end{minipage}
\end{center}
\thispagestyle{empty}
\end{titlepage}

\section{Introduction}
\label{sec:intro}

This paper deals with the subject of oversampling and undersampling
---the latter also known as aliasing in the engineering and signal processing
literature--- in the context of de Branges Hilbert spaces of entire functions
(dB spaces for short).
These notions play a prominent role in the theory of Paley-Wiener spaces
\cite{MR1473224,MR1270907}. Since Paley-Wiener spaces are leading
examples of dB spaces, questions related to oversampling and
undersampling in dB spaces emerge naturally.

Paley-Wiener spaces stem from the Fourier
transform of functions with given compact support centred at zero, \viz,
\begin{equation*}
  \cPW_a:=\left\{f(z)=\int_{-a}^a e^{-ixz}\phi(x)dx : \phi\in L_2(-a,a)\right\}.
\end{equation*}
By the Whittaker-Shannon-Kotel'nikov theorem, any function
$f(z)\in\cPW_a$ is decomposed as follows.
\begin{equation}
\label{eq:wsk-sampling}
  f(z)=\sum_{n\in\Z}f\left(\frac{n\pi}{a}\right)\mathcal{G}_a\left(z,\frac{n\pi}{a}\right),
\qquad\mathcal{G}_a\left(z,t\right):=\frac{\sin\left[a(z-t)\right]}{a(z-t)},
\end{equation}
where the convergence of the series is uniform in any compact subset of
$\C$. The function $\mathcal{G}_a\left(z,t\right)$ is referred to as
the sampling kernel.

In oversampling, the starting point is a function
$f(z)\in\cPW_a\subset\cPW_b$ ($a<b$).  Then, in addition to
\eqref{eq:wsk-sampling}, one has
\begin{equation*}
  f(z)=\sum_{n\in\Z}f\left(\frac{n\pi}{b}\right)
\mathcal{G}_b\left(z,\frac{n\pi}{b}\right)
\end{equation*}
Moreover, $f(z)$ admits a different representation
\begin{equation}
  \label{eq:wsk-oversampling}
   f(z)=\sum_{n\in\Z}f\left(\frac{n\pi}{b}\right)
  	\widetilde{\mathcal{G}}_{ab}\left(z,\frac{n\pi}{b}\right),
\end{equation}
with a modified sampling kernel $\widetilde{\mathcal{G}}_{ab}(z,t)$ depending on
$a$ and $b$ (see \cite[Thm. 7.2.5]{MR1473224}).  While the convergence
of the sampling formula \eqref{eq:wsk-sampling} is unaffected by $l_2$
perturbations of the samples $f\left(\frac{n\pi}{a}\right)$, formula
\eqref{eq:wsk-oversampling} is more robust because it is convergent
even under $l_\infty$ perturbations of the samples.  That is, if the
sequence $\{\epsilon_n\}_{n\in\Z}$ is bounded and one defines
\begin{equation}
 \label{eq:wsk-oversampling-with-error}
\widetilde{f}(z):=\sum_{n\in\Z}\left[f\left(\frac{n\pi}{b}\right)+\epsilon_n\right]
  	\widetilde{\mathcal{G}}_{ab}\left(z,\frac{n\pi}{b}\right),
\end{equation}
then $\abss{f(z)-\widetilde{f}(z)}$ is uniformly bounded in compact subsets
of $\C$ \cite[Thm. 7.2.5]{MR1473224}.

Undersampling, on the other hand, looks for the approximation of a
function $f(z)$ in
$\cPW_b\setminus\cPW_a$ by another one formally constructed using the sampling
formula \eqref{eq:wsk-sampling}, namely,
\begin{equation}
\label{eq:wsk-undersampling}
  \widehat{f}(z)
  	=\sum_{n\in\Z}f\left(\frac{n\pi}{a}\right)
   \mathcal{G}_a\left(z,\frac{n\pi}{a}\right).
\end{equation}
The series in \eqref{eq:wsk-undersampling} is indeed convergent and,
moreover,
$\abss{f(z)-\widehat{f}(z)}$ is uniformly bounded in compact subsets
of $\C$. Formula \eqref{eq:wsk-undersampling} yields in fact an
approximation not only for functions in $\cPW_b\setminus\cPW_a$, but
for the Fourier transform of elements in $L_1(\R)\cap L_2(\R)$
\cite[Thm.~7.2.9]{MR1473224}.

Oversampling and undersampling are, to some extent, consequences of
the fact that the chain of Paley-Wiener spaces $\cPW_s$,
$s\in(0,\infty)$, is totally ordered by inclusion.  As this is a
property shared by all dB spaces in the precise sense of
\cite[Thm.\,35]{MR0229011}, it is expected that analogous notions
should make sense in this latter class of spaces. We note that
sampling formulas generalizing \eqref{eq:wsk-sampling} are known for
arbitrary reproducing kernel Hilbert spaces (see e.g. Kramer-type
formulas in \cite{MR3275436,jorgensen-2016,MR2345302,MR3011977}), dB
spaces among them.  Analysis of error due to noisy samples and
aliasing, among other sources, in Paley-Wiener spaces goes back at
least to \cite{papoulis-1966}.  More recent literature on the subject
is, for instance,
\cite{bailey-2015,bailey-2016,bodmann-2012,kozachenko-2016}. However,
to the best of our knowledge, estimates for oversampling and
undersampling are not known for dB spaces apart from the Paley-Wiener
class.

A function $f(z)$ belonging to a dB space $\cB$ obviously admits a
representation in terms of an orthogonal basis. In particular,
\begin{equation}
 \label{eq:simple-sampling-db}
  f(z)=\sum_{t\in\spec(S(\gamma))}f(t)\frac{k(z,t)}{k(t,t)} ,
\end{equation}
where $k(z,w)$ is the reproducing kernel of $\cB$ and $S(\gamma)$ is a
canonical selfadjoint extension of the operator of multiplication by
the independent variable in $\cB$. The expansion
(\ref{eq:simple-sampling-db}) is a sampling formula with
$k(z,t)/k(t,t)$ being its sampling kernel. Note that
(\ref{eq:wsk-sampling}) is a particular realization of
(\ref{eq:simple-sampling-db}) for the dB space $\cPW_a$.

In order to obtain oversampling and undersampling estimates in analogy
to the Paley-Wiener case, we look into dB spaces of the form
\begin{equation}
\label{eq:our-db-spaces}
\cB_s = \left\{f(z) = \int_0^s \xi(x,z)\phi(x)\, dx : \phi\in L_2(0,s)\right\},
\end{equation}
where $\xi(x,z)$ solves
\begin{equation*}
  -\frac{d^2}{dx^2}\varphi+V(x)\varphi=z\varphi,\quad x\in (0,s),\quad z\in\C,
\end{equation*}
for some $s\in(0,\infty)$ and with Neumann boundary condition at $x=0$ (see
Section~\ref{sec:prel}). Here $V\in L_1(0,s)$ is a real function.
By construction $\cB_s\subset\cB_{s'}$ whenever $s<s'$ (for more on
this, see \cite{MR1943095}).

Define
\begin{equation*}
  \mathcal{K}_s(z,t):=\frac{k_s(z,t)}{k_s(t,t)}\,,
\end{equation*}
where $k_s(z,w)$ is the reproducing kernel of the space $\cB_s$. If
$S_s(\gamma)$ is a selfadjoint extension of the multiplication
operator in $\cB_s$, then any $f(z)\in\cB_s$ has the representation
\begin{equation*}
  f(z)=\sum_{t\in\spec(S_s(\gamma))}f(t)\mathcal{K}_s\left(z,t\right)\,.
\end{equation*}

Our main results are Theorems~\ref{thm:main-oversampling} and
\ref{thm:main-subsampling}, which can be
summarized as follows:

\begin{thm}[oversampling]
Assume that $V$ is real-valued and in $\text{\rm AC}[0,\pi]$ (the set
of absolutely continuous functions in $[0,\pi]$). Consider an
arbitrary $f(z)\in\cB_a$, where $a\in(0,\pi)$.
For a given $\{\epsilon_t\}\in l_\infty$, define
\begin{equation*}
  \widetilde{f}(z):=
  \sum_{t\in\spec(S_{\pi}(\pi/2))}\widetilde{\mathcal{K}}_{a\pi}(z,t)
  \left(f(t)+\epsilon_t\right),
\end{equation*}
where $\widetilde{\mathcal{K}}_{ab}(z,t)$ is given in \eqref{eq:def-K}.
Then, for every compact set $K$ of $\C$, there is a
constant $C(a,K,V)>0$ such that
\[
\abs{f(z)-\widetilde{f}(z)} \le
C(a,K,V)\norm{\epsilon}_\infty,\quad z\in K.
\]
\end{thm}
We remark that the bound is uniform for $f(z)\in\cB_a$. Note that
$\widetilde{\mathcal{K}}_{ab}(z,t)$ is a modified sampling
kernel analogous to the one in (\ref{eq:wsk-oversampling-with-error}).
\begin{thm}[undersampling]
Assume $V$ is real-valued and in $\text{\rm AC}[0,b]$ with $b>\pi$.
Given $g(z)\in\cB_b\setminus\cB_\pi$, define
\begin{equation*}
  \widehat{g}(z):=\sum_{t\in\spec(S_\pi(\pi/2))}g(t)
\mathcal{K}_\pi\left(z,t\right)\,.
\end{equation*}
Then, for each compact set $K\subset\C$, there is a constant $D(b,K,V)>0$ such that
\[
\abs{g(z)-\widehat{g}(z)} \le D(b,K,V)\int_{\pi}^b\abs{\psi(x)}dx
\]
uniformly on $K$, where $\psi\in L_2(0,b)$ obeys
$g(z)=\inner{\xi(\cdot,\cc{z})}{\psi(\cdot)}_{L_2(0,b)}$.
\end{thm}

These results are somewhat limited in several respects. First, we show
oversampling relative to the pair $\cB_a\subset\cB_\pi$, and
undersampling relative to the pair $\cB_\pi\subset\cB_b$ (for dB
spaces defined according to \eqref{eq:our-db-spaces}). These
particular choices are related to a convenient simplification in the
proofs, but our results can be extended to an arbitrary pair
$\cB_a\subset\cB_b$ by a scaling argument. Second, the sampling
formulae use the spectra of selfadjoint operators with Neumann
boundary condition at the left endpoint. This choice simplifies the
asymptotic formulae for eigenvalues of the associated Schr\"odinger
operator; it can also be removed but at the expense of a somewhat
clumsier analysis. In our opinion this extra workload would not add
anything substantial to the results. Finally, and more importantly
from our point of view, our assumption on the potential functions is a
bit too restrictive. In view of \cite{MR1943095}, we believe that our
results should be valid just requiring $V\in L_1(0,s)$, but relaxing
our present assumption on $V$ would require some major changes in the
details of our proofs. Further generalizations of the results
presented here (in particular, involving a wider class of dB spaces)
are the subject of a future work.

About the organization of this work: Section~\ref{sec:prel} recalls
the necessary elements on de Branges spaces and regular Schr\"odinger
operators.  Section~\ref{oversampling} deals with
oversampling. Undersampling is treated in
Section~\ref{undersampling}. The Appendix contains some technical
results.

\section{dB spaces and Schr\"{o}dinger operators}
\label{sec:prel}
There are various ways of defining a de Branges space (see
\cite[Sec. 19]{MR0229011}, \cite[Sec. 2]{MR1943095},
\cite{zbMATH06526214}). We recall the following definition: a Hilbert
space of entire functions $\cB$ is a de Branges (dB space) when it has
a reproducing kernel $k(z,w)$ and is isometrically invariant under the mappings
$f(z)\mapsto f^\#(z):=\cc{f(\cc{z})}$ and
\begin{equation*}
f(z)\mapsto \left(\frac{z-\cc{w}}{z-w}\right)^{{\rm Ord}_w(f)}f(z)\,,
\qquad w\in\C\,,
\end{equation*}
where ${\rm Ord}_w(f)$ is the order of $w$ as a zero of $f$.
The class of dB spaces appearing in this work has the following additional
properties:

\begin{enumerate}[label={\bf (a\arabic*)}]
\item Given any real point $x$, there is a function $f\in\cB$  such that
	$f(x)\ne 0$. \label{it:hb-without-realzeros}

\item $\cB$ is regular, i.\,e., for any $w\in\C$ and $f\in\cB$,
	$(z-w)^{-1} \left( f(z)-f(w) \right) \in \cB$. \label{it:regulatr-db}
\end{enumerate}

A distinctive structural property of dB spaces is that the set of dB
subspaces of a given dB space is totally ordered by inclusion
\cite[Thm. 35]{MR0229011}. For regular dB spaces (in the sense of
\ref{it:regulatr-db}) this means that, if $\cB_1$ and $\cB_2$ are
subspaces of a dB space that are themselves dB spaces, then either
$\cB_1\subset\cB_2$ or $\cB_1\supset\cB_2$ \cite[Sec. 6.5]{MR0448523}.

The operator $S$ of multiplication by the independent variable in a dB
space $\cB$ is defined by
\begin{equation}
  \label{eq:definition-multiplication-operator}
  (Sf)(z)=zf(z),\quad \dom(S):=\{f\in\cB: Sf\in\cB\}.
\end{equation}
This
operator is closed, symmetric and has deficiency
indices $(1,1)$.

In view of \ref{it:hb-without-realzeros}, the spectral
core of $S$ is empty (\cf \cite[Sec. 4]{MR1664343}), i.\,e., for any
$z\in\C$, the operator $(S-z I)^{-1}$ is bounded although, as a
consequence of the indices being $(1,1)$, its domain has codimension
one.  We consider dB spaces such that $S$ is densely defined and denote by $S(\gamma)$,
$\gamma\in [0,\pi)$, the selfadjoint restrictions of $S^*$.

Since
$\displaystyle \inner{(S^*-w)k(\cdot,\cc{w})}{f(\cdot)}
	= \inner{k(\cdot,\cc{w})}{(S - \cc{w})f(\cdot)}
	= 0$
for all $f(z)\in\dom(S)$, we have
$k(z,\cc{w})\in\ker(S^*-w I)$ for any $w\in\C$. Thus
\begin{equation}
  \label{eq:kernel-basis}
\left\{k(z,t) : t\in\spec(S(\gamma))\right\}\text{ is an orthogonal basis},
\end{equation}
where $\spec(S(\gamma))$ denotes the spectrum of $S(\gamma)$.
Hence, the sampling formula
\begin{equation}
\label{1ra formula kramer}
  f(z)=\sum_{t\in\spec(S(\gamma))}f(t)\frac{k(z,t)}{k(t,t)},
  \qquad f\in\cB,
\end{equation}
holds true. The convergence of this series is in the dB space, which
in turn implies uniform convergence in compact subsets of $\C$.

The dB spaces under consideration in this work are related to symmetric operators
arising from regular Schr\"odinger differential expressions. The
construction is similar to the one developed in \cite{MR1943095},
although there are other ways of generating dB spaces from
differential equations of the Sturm-Liouville type \cite{MR0276741}.

Consider a differential expression of the form
\begin{equation*}
\tau:=-\frac{d^2}{dx^2}+V(x),
\end{equation*}
where we assume
\begin{enumerate}[label={\bf (v\arabic*)}, series=onpottential]
\item \label{1ra propiedad del potencial} $V$ is real-valued and belongs to
	$L_1(0,s)$ for arbitrary $s>0$.
\end{enumerate}
For each $s>0$, $\tau$ determines a closed symmetric
operator $H_s$ in $L_2(0,s)$,
\begin{equation*}
\begin{gathered}
 \dom(H_s) :=\{\varphi\in L_2(0,s): \tau\varphi\in L_2(0,s),
 	\varphi'(0)=\varphi(s)=\varphi'(s)=0 \}
 \\[1mm]
    H_s\varphi :=\tau\varphi .
\end{gathered}
\end{equation*}
This operator is known to have deficiency indices $(1,1)$ and empty
spectral core, that is,
\begin{equation*}
  \left\{z\in\C:\text{there is } C_z>0\text{ such that }\norm{(H_s-z
      I)\varphi}\ge C_z\norm{\varphi}
    \right\}=\C.
\end{equation*}

The selfadjoint extensions of $H_s$ are given by
\begin{equation}
\label{selfadjoint extensions}
\begin{gathered}
 \dom\left(H_s(\gamma)\right)
 	:=\left\{\begin{gathered}\varphi\in L_2(0,s): \tau\varphi\in L_2(0,s),
 	   \\
       \varphi'(0)=0,\ \varphi(s)\cos{\gamma}+\varphi'(s)\sin{\gamma}=0
         \end{gathered}\right\}
 \\[1mm]
 H_s(\gamma) \, \varphi :=\tau\varphi,
\end{gathered}
\end{equation}
with $\gamma\in[0,\pi)$.
Finally, the adjoint operator of $H_s$ is
\begin{equation*}
  \dom(H_s^*)  :=\left\{\varphi\in L_2(0,s): \tau\varphi\in L_2(0,s)\,,
    				   \varphi'(0)=0 \right\},
  \qquad
  H_s^*\varphi :=\tau\varphi.   \label{adjoint of H}
\end{equation*}

Let $\xi:\R_+\times\C\to\C$ be the solution of the eigenvalue problem
\begin{equation*}
\tau\xi(x,z)=z\xi(x,z), \quad  \xi(0,z)=1, \quad   \xi'(0,z)=0.
\end{equation*}
(The derivative is taken with respect to the first argument.) The
function $\xi(x,z)$ is real entire for any fixed $x\in\R_+$
\cite[Thm.\,1.1.1]{MR1136037}, \cite[Thm.\,9.1]{MR2499016}.  Also,
$\xi(\cdot,z) \in L_2(0,s)$ for any $z\in\C$.  Using
\cite[Sec.\,4]{zbMATH06526214} one then establishes that
$\xi(\cdot,z)$ is entire as an $L_2(0,s)$-valued map.  Note that
$\xi(\cdot,z)$ depends on the potential $V$ but does not depend on the
right endpoint $s$.

According to \cite[Props.\,2.12 and 2.14]{MR3002855}
\cite[Thm.\,16]{zbMATH06526214}, the functions
\begin{equation}
  \label{eq:map-to-dB}
  f(z)=\inner{\xi(\cdot,\cc{z})}{\varphi(\cdot)}_{L_2(0,s)},
\end{equation}
with $\varphi\in L_2(0,s)$, form a dB space $\cB_s$ with the norm
given by
\begin{equation}   \label{norma en Bs}
\norm{f}_{\cB_s}=\norm{\varphi}_{L_2(0,s)}.
\end{equation}
A straightforward computation shows that the reproducing kernel of
$\cB_s$ is
\begin{equation}
  \label{eq:reproducing-xi}
  k_s(z,w)=\inner{\xi(\cdot,\cc{z})}{\xi(\cdot,\cc{w})}_{L_2(0,s)}.
\end{equation}
\begin{remark}
  \label{rem:decomposition-xi}
In view of \eqref{eq:reproducing-xi}, $k_s(z,w)$ and $\xi(\cdot,\cc{w})$ are related by the isometry
\eqref{eq:map-to-dB}. Hence, using  \eqref{eq:kernel-basis} and
expression \eqref{norma en Bs} for the norm in $\cB_s$, one obtains
\begin{equation}
   \label{eq:decomposition-function-L2}
   \varphi(x) =
   \sum_{t\in\spec(H_s(\gamma))}
     \frac{1}{k_s(t,t)} \inner{\xi(\cdot,t)}{\varphi(\cdot)}_{L_2(0,s)}
     \xi(x,t), \qquad \varphi \in L_2(0,s) ,
\end{equation}
where the series converges in the $L_2$-norm.
\end{remark}

If $r<s$, then $\cB_{r}$ is a proper dB subspace of $\cB_s$. Indeed,
$\{\cB_{r}:r\in(0,s)\}$ is a chain of dB subspaces of $\cB_s$ in
accordance with \cite[Thm.~35]{MR0229011}. The isometry from
$L_2(0,s)$ onto $\cB_s$ induced by \eqref{eq:map-to-dB} transforms
$H_s$ into the operator of multiplication by the independent variable
in $\cB_s$ (see \eqref{eq:definition-multiplication-operator}), the
latter will subsequently be denoted by $S_s$. Also, the selfadjoint
extensions $H_s(\gamma)$ are transformed into the selfadjoint
extensions $S_s(\gamma)$ of $S_s$. When referring to unitary
invariants (such as the spectrum), we use interchangeably either
$H_s(\gamma)$ or $S_s(\gamma)$ throughout this text.

\begin{remark}
  The space $\cB_s$ constructed from $L_2(0,s)$ via
  \eqref{eq:map-to-dB} depends on the potential $V$, which is assumed
  to satisfy \ref{1ra propiedad del potencial}.  However, as shown in
  \cite[Thm.~4.1]{MR1943095}, the set of entire functions in $\cB_s$
  is the same for all $V\in L_1(0,s)$; what changes with $V$ is the
  inner product in $\cB_s$. Noteworthily, since the operator $S_s$ of multiplication by
  the independent variable is defined in its maximal domain (see
  \eqref{eq:definition-multiplication-operator}), it has always the
  same domain and range and acts in the same way; yet, by modifying
  the metric of the space, each $V\in
  L_1(0,s)$ gives rise to a different family of selfadjoint extensions
  of $S_s$. As a consequence, every function in $\cB_s$ can be sampled
  by (\ref{1ra formula kramer}) using any sequence $\{\lambda_n\}$ as
  sampling points, as long as there exists $V\in L_1(0,s)$ such that
  $\{\lambda_n\}$ is the spectrum of some selfadjoint extension of the
  corresponding operator $H_s$. This fact can be considered as a
  generalization of the notion of irregular sampling, quite well
  studied in Paley-Wiener spaces by means of classical analysis; the
  Kadec's 1/4 Theorem is a chief example of this kind of results \cite{MR0162088}.
\end{remark}

\section{Oversampling} \label{oversampling}

The oversampling of a function in
$\cB_a$ is related to the fact that it can be sampled as a function in
$\cB_b$ and the sampling kernel can be modified in such a way that the
sampling series is convergent under $l_\infty$ perturbations of the
samples (see the Introduction).

Let $0<a<b<\infty$ and $V$ be as in \ref{1ra propiedad del potencial}.
Any $\varphi \in L_2(0,a)$ can be identified with an element in
$L_2(0,b)$ since
\begin{equation}   \label{phi con caracteristicas}
\varphi = \varphi \chi_{[0,a]} + 0 \chi_{(a,b]} ,
\end{equation}
where $\chi_E$ denotes the characteristic function of a set $E$. Define
\begin{equation}
\label{def de funcion R}
\cR(x) = \cR_{ab}(x) \defeq  \chi_{[0,a]}(x) + \frac{b-x}{b-a}
\chi_{(a,b]}(x), \qquad x \in [0,b] .
\end{equation}
Taking into account \eqref{eq:decomposition-function-L2} with $s=b$,
(\ref{phi con caracteristicas}) and (\ref{def de funcion R}) imply
\begin{equation}
  \label{phi con r}
\varphi(x) =
\sum_{t \in \spec( H_b(\gamma) )}\frac{1}{k_b(t,t)}\inner{\xi(\cdot,t)}
  {\varphi(\cdot)}_{L_2(0,b)}\mathcal{R}(x)\xi(x,t) ,
\end{equation}
where the convergence is in $L_2(0,b)$.
Plugging \eqref{phi con r} into \eqref{eq:map-to-dB} with $s=b$,
we obtain
\begin{equation}
  \label{expresion de g}
f(z) = \sum_{t \in \spec( H_b(\gamma) )}\frac{1}{k_b(t,t)}
\inner{\xi(\cdot,\cc{z})}{\mathcal{R}(\cdot)  \xi(\cdot,t)}_{L_2(0,b)}
f (t) , \qquad z \in \C  ,
\end{equation}
which converges uniformly in compact subsets of $\C$.
\begin{hypothesis}
  \label{hyp:convergence-oversampling}
Given $0<a<b$, the series
\begin{equation}  \label{serie undersampling hipotesis}
\sum_{t\in\spec( H_b(\gamma) )}
\frac{1}{k_b(t,t)}
\abs{\inner{\xi(\cdot,\cc{z})}{\mathcal{R}_{ab}(\cdot) \xi(\cdot,t)}_{L_2(0,b)}}
\end{equation}
converges uniformly in compact subsets of $\C$.
\end{hypothesis}

Assume that Hypothesis~\ref{hyp:convergence-oversampling} is met. Enumerate any
given sequence $\epsilon \in l_\infty$ such that
$\epsilon=\{\epsilon_t\}_{t \in \spec( H_b(\gamma) )}$.
Define
\begin{equation}
  \label{eq:def-K}
  \widetilde{\mathcal{K}}_{ab}(z,t):=\frac{1}{k_b(t,t)}\inner{\xi(\cdot,\cc{z})}
  {\mathcal{R}_{ab}(\cdot) \, \xi(\cdot,t )}_{L_2(0,b)}.
\end{equation}
In view of
\eqref{expresion de g}, the function
\begin{equation}
\label{f con perturbaciones en los pntos de muestreo}
\widetilde{f}(z) \defeq \sum_{t \in \spec( H_b(\gamma) )}
\widetilde{\mathcal{K}}_{ab}(z,t)
   \left( f(t) + \epsilon_t \right) , \quad z \in \C,
 \end{equation}
 is well defined and the defining series converges uniformly in compact
subsets of $\C$.
Moreover,
\begin{equation*}
\abs{ \widetilde{f}(z) - f(z) }
	\le \norm{\epsilon}_{l_\infty}\!\!\!
	\sum_{t \in \spec( H_b(\gamma) )} \!\frac{1}{k_b(t,t)}
	\abs{\inner{\xi(\cdot,\cc{z})}
  	{\mathcal{R}(\cdot) \, \xi(\cdot,t )}_{L_2(0,b)}},
\end{equation*}
for all $z \in \C$. Thus, the difference $\abss{\widetilde{f}(z)-f(z)}$ is
uniformly bounded in compact subsets of $\C$. Below we prove that Hypothesis
\ref{hyp:convergence-oversampling} holds true when

\begin{enumerate}[label={\bf (v\arabic*)}, resume*=onpottential]
\item \label{propiedades requeridas para el potencial} $V$
is real-valued and in $\text{\rm AC}[0,b]$ (hence it satisfies \ref{1ra
  propiedad del potencial} for $s\le b$).
\end{enumerate}
This is performed in two stages, the first one deals with the
case $V \equiv 0$, the second one employs perturbative methods to
consider the general case.

If $V\equiv 0$, the function $\xi$ given in Section~\ref{sec:prel} is
\begin{equation}
  \label{eq:xi-without-potential}
  \xi(x,z) = \cos(\sqrt{z}\,x) , \qquad x\in\R_+\, .
\end{equation}
Whenever we refer to the function $\xi$ corresponding to $V\equiv 0$,
we write the right-hand-side of (\ref{eq:xi-without-potential}). We
reserve the use of the symbol $\xi$ only for the case $V\not\equiv 0$. Also, throughout
this paper we use the main branch of the square root function.

As mentioned in the Introduction, for the sake of simplicity we assume $b=\pi$ and
fix $\gamma=\pi/2$. A straightforward calculation yields
\begin{equation}
\label{espectro sin potencial}
\spec \left( H_\pi\left(\pi/2\right) \right) = \{n^2 \,:\, n \in \N\cup\{0\}\} .
\end{equation}
Moreover, by substituting \eqref{eq:xi-without-potential} into
\eqref{eq:reproducing-xi}, we verify that the reproducing kernel
$\accentset{\circ}{k}_\pi(z,w)$ corresponding to the case $V\equiv 0$ satisfies
\begin{equation}   \label{norma funciones propias al cuadrado}
  \accentset{\circ}{k}_\pi(n^2,n^2) =
  \begin{cases}
    \pi & \text{ if } n=0 ,
    \\
    \frac{\pi}{2} & \text{ if }n\in\N .
  \end{cases}
\end{equation}
In the remainder of this section, we denote $\inner{\cdot}{\cdot}_{L_2(0,\pi)}$ simply as
$\inner{\cdot}{\cdot}$.

\begin{proposition}
\label{hipotesis sobremuestreo sin potencial}
Hypothesis~\ref{hyp:convergence-oversampling} holds true under the
assumption $V \equiv 0$, $b=\pi$, and $\gamma=\pi/2$.
\end{proposition}

\begin{proof}
  Consider a compact set $K$ in $\C$ such that
  $\spec(H_\pi(\pi/2))$ intersects $K$ only at one point $n_0^2$ with
  $n_0\in\N$. It will be clear at the end of the proof that there is
  no loss of generality in this assumption. First note that
  $\abs{\inner{\cos(\sqrt{\cc{z}}\,\cdot)}{\cR(\cdot)\cos(n_0\,\cdot)}}$
  is uniformly bounded in $K$ (one can use the Cauchy-Schwarz
  inequality and note that the factor depending on $z$ is continuous
  in $K$). On the other hand, by Lemma \ref{valor
    importante del producto interior}, \small
  \begin{align*}
&\sum_{n\ne n_0}
\abs{
\inners{\cos(\sqrt{\cc{z}}\,\cdot)}{\cR(\cdot)\cos(n\,\cdot)}}\nonumber
\\[1mm]
&=\frac{1}{2}\sum_{n\ne n_0}
\abs{
	\frac{\cos ( (\sqrt{z}+n)a )-(-1)^n \cos(\sqrt{z}\pi)}{(\pi-a)(\sqrt{z}+n)^2}
	+ \frac{\cos ( (\sqrt{z}-n)a )
	-(-1)^n \cos(\sqrt{z}\pi)}{(\pi-a) (\sqrt{z}-n)^2}
    }\nonumber
\\[1mm]
&\le \frac{e^{\pi\abs{\im\sqrt{z}}}}{(\pi-a)}
	\sum_{n\ne n_0}
	\left(\frac{1}{\abs{\sqrt{z} + n}^2} + \frac{1}{\abs{\sqrt{z} -n}^2}\right).
\end{align*}
\normalsize
Thus, taking into account \eqref{norma funciones propias al cuadrado},
the series \eqref{serie undersampling hipotesis}
converges uniformly in $K$.
\end{proof}

Now, let us address the case of non-zero $V$ satisfying
\ref{propiedades requeridas para el potencial}. As before we set $b=\pi$ and
$\gamma=\pi/2$. Also, we assume $\spec(H_\pi(\pi/2)) = \{\lambda_n\}_{n=0}^\infty$
ordered such that $\lambda_{n-1} < \lambda_n$ for all $n \in \N$.
The subsequent analysis make use of the following auxiliary functions.

\begin{definition}   \label{funciones auxiliares}
For each $x \in [0,\pi]$, $n \in \N$ and $z \in \C$, consider
\begin{gather}
\rho(x) \defeq \frac{1}{2} \int_0^x V(y) dy
- \frac{x}{2 \pi} \int_0^\pi V(y) dy ,\nonumber
\\
T(x,n) \defeq \xi(x,\lambda_n)-\cos(nx)-
\frac{\rho(x)}{n} \sin(nx) ,\label{eq:T-def}
\\[1.5mm]
F(x,z) \defeq \xi(x,z) - \cos(\sqrt{z} \, x) .\nonumber
\end{gather}
\end{definition}

\begin{lemma}
\label{sumandos oversampling}
Let $V$ be as in \ref{propiedades requeridas para el potencial} with $b=\pi$.
There exists $N \in \N$ such that, if
$n \ge N$, then
\small
\begin{equation*}
\abs{\inners{\xi(\cdot,\cc{z})}{\cR(\cdot)\xi(\cdot,\lambda_n)}
   - \inners{\cos(\sqrt{\cc{z}}\,\cdot)}{\cR(\cdot)\cos(n\,\cdot)}}
\le C_\pi\frac{e^{\abs{\im\sqrt{z}}\pi}}{n^2}
\left( 1 + \frac{1 + \abs{z}}{1+\pi\abs{z}^{1/2}} \right)
\end{equation*}
\normalsize
for every $z \in \C$. Here $C_\pi$ is a positive number depending on $V$.
\end{lemma}
\begin{proof}
In terms of the functions introduced in Definition \ref{funciones auxiliares},
one writes
  \begin{multline}\label{eq:equality-with-terms}
  \big\langle\xi(\cdot,\cc{z}),\cR(\cdot)\xi(\cdot,\lambda_n)\big\rangle
    -
   \inners{\cos(\sqrt{\cc{z}}\,\cdot)}{\cR(\cdot)\cos(n\,\cdot)}
\\[2mm]
  =
    \int_0^\pi \Big[ \cos(\sqrt{z}x)\cR(x)\frac{\rho(x)}{n}\sin(nx)
    + F(x,z)\cR(x)\frac{\rho(x)}{n}\sin(nx)
\\[2mm]
    + F(x,z)\cR(x)\cos(nx)
    + \cos(\sqrt{z}x)\cR(x)T(x,n)
  \\[2mm]  + F(x,z)\cR(x)T(x,n)\Big] dx .
  \end{multline}
It will be shown that each of the five terms on the right-hand side of
\eqref{eq:equality-with-terms} is appropriately bounded. For the first
term, one uses the inequality (\ref{eq:int2}) of
Lemma~\ref{lem:integrals-1-3} and the first inequality of
Lemma~\ref{int 2 oversam}. The estimate of the second term is obtain
by combining (\ref{eq:int3}) of Lemma~\ref{lem:integrals-1-3} and the
second inequality of Lemma~\ref{int 2 oversam}. The third term on the right-hand side of
\eqref{eq:equality-with-terms} is estimated in Lemma~\ref{int 1
  oversam}.

As regards the fourth and fifth terms in
\eqref{eq:equality-with-terms}, one proceeds as follows.
From Lemma~\ref{lemma:nuevo}\ref{eq:asymp-T}, it follows that
\begin{equation*}
  \abs{T(x,n)}\le \frac{D}{n^2},\quad D>0,
\end{equation*}
uniformly with respect to $x\in [0,\pi]$ for $n$ sufficiently
large. Also, $\abs{\mathcal{R}(x)}\le 1$ according to \eqref{def de
  funcion R}. Therefore, one has
\begin{equation}
  \label{eq:first-A7}
  \abs{\int_0^\pi T(x,n) \cR(x) \cos(\sqrt{z} x) dx} \le
    \frac{C_1}{n^2}\, e^{\abs{\im\sqrt{z}}\pi}
\end{equation}
since
\begin{equation*}
  \abss{\cos(\sqrt{z}x)}\le \exp(\abss{\im\sqrt{z}}\pi),\quad x\in [0,\pi].
\end{equation*}
The bound for the remaining term follows by a similar reasoning taking into
account \eqref{estimado funciones propias con potencial}. Thus,
\begin{equation}
  \label{eq:second-A7}
  \abs{\int_0^\pi T(x,n) \cR(x) F(x,z) dx} \le
    \frac{C_2}{n^2} \frac{\pi^2}{1+\pi\abs{z}^{1/2}}\,e^{\abs{\im\sqrt{z}}\pi} .
\end{equation}
By combining the estimates of the first three terms, together with
\eqref{eq:first-A7} and \eqref{eq:second-A7}, the bound of the
statement is established.
\end{proof}

\begin{proposition}
\label{hipotesis sobremuestreo caso general}
Let $V$ be as in \ref{propiedades requeridas para el potencial}. If
$b=\pi$ and $\gamma=\pi/2$, then Hypothesis~\ref{hyp:convergence-oversampling}
holds true.
\end{proposition}
\begin{proof}
From Lemma~\ref{lemma:nuevo}\ref{eq:asymp-k} we know that
$k_\pi(\lambda_n,\lambda_n)-\accentset{\circ}{k}_\pi(n^2,n^2)=\cO(n^{-2})$ as
$n \to \infty$. This implies that
\begin{equation*}
  k_\pi(\lambda_n,\lambda_n)\ge \frac{\accentset{\circ}{k}_\pi(n^2,n^2)}{2}=\frac{\pi}{4}
\end{equation*}
for $n$ suficiently large, where we have used \eqref{norma funciones propias al cuadrado}. Hence,
\begin{equation*}
\abs{\frac{1}{k_\pi(\lambda_n,\lambda_n)} - \frac{1}{\accentset{\circ}{k}_\pi(n^2,n^2)}}
=\frac{\abs{k_\pi(\lambda_n,\lambda_n)-\frac{\pi}2}}{\frac{\pi}2k_\pi(\lambda_n,\lambda_n)}
\le\frac{8}{\pi^2}\abs{k_\pi(\lambda_n,\lambda_n)-\frac{\pi}2}
\end{equation*}
for $n$ suficiently large. Again resorting to
Lemma~\ref{lemma:nuevo}\ref{eq:asymp-k}, one obtains
\begin{equation}   \label{inverso mult normas eigen funciones}
\frac{1}{k_\pi(\lambda_n,\lambda_n)} - \frac{1}{\accentset{\circ}{k}_\pi(n^2,n^2)}
= \cO(n^{-2})   ,  \quad n \to \infty    .
\end{equation}
Due to Lemma \ref{sumandos oversampling} and
(\ref{inverso mult normas eigen funciones}) there exists $N\in \N$
such that, if $n \ge N$, then
\begin{equation*}
\abs{\frac{\inners{\xi(\cdot,\cc{z})}{\cR(\cdot)\xi(\cdot,\lambda_n)}}
	{k_\pi(\lambda_n,\lambda_n)} -
	\frac{\inner{\cos(\sqrt{\cc{z}}\,\cdot)}{\cR(\cdot)\cos(n\,\cdot)}}
	{\accentset{\circ}{k}_\pi(n^2,n^2)}}
 	\le
	\frac{c_1(z)}{n^2} ,
\end{equation*}
for all $z \in \C$, and where $c_1:\C\to\R$ is a positive continuous function.
As a consequence of the previous inequality, there exists another positive continuous
function $c_2:\C\to\R$ such that
\begin{equation*}
\sum_{n=0}^\infty
\abs{\frac{\inners{\xi(\cdot,\cc{z})}{\cR(\cdot)\xi(\cdot,\lambda_n)}}
	{k_\pi(\lambda_n,\lambda_n)} -
	\frac{\inner{\cos(\sqrt{\cc{z}}\,\cdot)}{\cR(\cdot)\cos(n\,\cdot)}}
	{\accentset{\circ}{k}_\pi(n^2,n^2)}}
	\le
	c_2(z) .
\end{equation*}
Hence, by Proposition~\ref{hipotesis sobremuestreo sin potencial},
the series \eqref{serie undersampling hipotesis}
converges uniformly in compact subsets of $\C$.
\end{proof}

Arguing as in the paragraph below
Hypothesis~\ref{hyp:convergence-oversampling}, one arrives at the
following assertion in which the oversampling procedure is established
(see the Introduction).

\begin{theorem}
\label{thm:main-oversampling}
Suppose $V$ obeys \ref{propiedades requeridas para el potencial} with $b=\pi$.
Consider $\cB_a$ with $a\in(0,\pi)$. Then, for every compact set $K\subset\C$, there
exist a constant $C(a,K,V)>0$ such that
\[
\abs{ f(z)-\widetilde{f}(z) } \le C(a,K,V)\norm{\epsilon}_\infty, \quad z\in K,
\]
for all $f(z)\in\cB_a$, where $\epsilon=\{\epsilon_t\}$ is any bounded
real sequence and $\widetilde{f}(z)$ is given by \eqref{f con
  perturbaciones en los pntos de muestreo} with $b=\pi$ and $\gamma=\pi/2$.
\end{theorem}

\section{Undersampling} \label{undersampling}

In this section, we treat undersampling of functions in
$\mathcal{B}_b\setminus\mathcal{B}_a$ ($a<b$) with the sampling points given
by the spectrum of $S_a(\gamma)$ as explained in the Introduction.
\begin{hypothesis}
  \label{hyp:absolute-convergence-subsampling}
  For $a<b$ and each $z \in \C$, the series
\begin{equation}
 \label{serie hipotesis undersampling}
  \sum_{t\in\spec(H_a(\gamma))} \frac{k_a(t,\cc{z})}{k_a(t,t)} \xi(x,t)
\end{equation}
converges absolutely and uniformly with respect to $x \in
[0,b]$.
\end{hypothesis}

\begin{remark}
  \label{rem:convergence-l2-and-one-term}
Note that \eqref{eq:reproducing-xi} and
\eqref{eq:decomposition-function-L2} imply that the series
\begin{equation}
\label{coef de fourier de
xi}
\sum_{t\in\spec(H_a(\gamma))} \frac{k_a(t,\cc{z})}{k_a(t,t)}
\xi(\cdot,t)
\end{equation}
converges to $\xi(\cdot,z)$ in $L_2(0,a)$ for each $z\in\C$. Due to
\eqref{eq:kernel-basis}, if $z=\lambda\in\spec(H_a(\gamma))$, then
$k_a(t,\lambda)=0$ for
$t\in\spec(H_a(\gamma))\setminus\{\lambda\}$, in which case 
the series \eqref{coef de fourier de xi} and
\eqref{serie hipotesis undersampling} have only one term.
\end{remark}

\begin{lemma}
\label{propiedades de xi extendida}
Assume that Hypothesis~\ref{hyp:absolute-convergence-subsampling} is met. Define
\begin{equation*}
\xi^{ext}_a(x,z) \defeq \sum_{t\in\spec(H_a(\gamma))} \frac{k_a(t,\cc{z})}{k_a(t,t)}
\xi(x,t)  , \quad  x \in [0,b] , \quad z \in\C  \,.
\end{equation*}
Then, for each $z\in\C$,
\begin{enumerate}[label={(\roman*)}]
\item $\xi^{ext}_a(\cdot,z)$ is continuous in $[0,b]$,
	\label{lem:xi-ext-is-continuous}

\item $\xi^{ext}_a(x,z) = \xi(x,z)$ for a. e. $x \in [0,a]$, and \label{lem:xi-ext-is-xi}
\item the function
	$h_a(z):=\displaystyle\sup_{x \in [a,b]} \abss{\xi^{ext}_a(x,z) - \xi(x,z)}$
	is continuous in $\C$. \label{lem:h-is-continuous}
\end{enumerate}
Moreover,
\begin{enumerate}[resume*]
\item if $\psi \in L_2(0,b)$ and $g(z)\in \cB_b$ are related by the
	isometry \eqref{eq:map-to-dB}, then
	\begin{equation}   \label{xi ext prod int}
	\inner{ \xi^{ext}_a(\cdot,\cc{z}) }{ \psi(\cdot) }_{L_2(0,b)}  =
	\sum_{t\in\spec(H_a(\gamma))} \frac{k_a(t,\cc{z})}{k_a(t,t)}  g(t) ,
	\qquad z \in \C .
	\end{equation}
	\label{lem:xi-ext-inner-psi}
\end{enumerate}
\end{lemma}
\begin{proof}
Enumerate $\spec\left(H_a\left(\gamma\right)\right)=\{\lambda_n\}_{n=0}^\infty$
such that $\lambda_{n-1} < \lambda_n$ for all $n \in \N$.
Then \ref{lem:xi-ext-is-continuous} is a straightforward consequence of
Hypothesis~\ref{hyp:absolute-convergence-subsampling}.
Due to \ref{lem:xi-ext-is-continuous}, $\xi^{ext}_a(\cdot,z)$ is an element of $L_2(0,a)$
for each $z\in\C$. Thus, Hypothesis~\ref{hyp:absolute-convergence-subsampling} implies
 \begin{equation*}
   \lim_{m \to \infty}
   \norm{ \xi^{ext}_a(\cdot,z) -
         \sum_{n=0}^m \frac{k_a(\lambda_n,\cc{z})}{k_a(\lambda_n,\lambda_n)}
         \xi(\cdot,\lambda_n)}_{L_2(0,a)} = 0 .
\end{equation*}
This, along with Remark~\ref{rem:convergence-l2-and-one-term}, yields
\ref{lem:xi-ext-is-xi}. Item \ref{lem:h-is-continuous}
follows from Lemma~\ref{lemma:analyis-exercise}. To prove \ref{lem:xi-ext-inner-psi},
apply the dominated convergence theorem, which holds because of
Hypothesis \ref{hyp:absolute-convergence-subsampling},
\[
\inner{ \xi^{ext}_a(\cdot,\cc{z}) }{ \psi(\cdot) }_{L_2(0,b)}
	= \lim_{m \to \infty} \sum_{n=0}^m \frac{k_a(\lambda_n,\cc{z})}
	{k_a(\lambda_n,\lambda_n)} \int_0^{b} \xi(x,\lambda_n) \psi(x) \,dx .\qedhere
\]
\end{proof}

Assume that Hypothesis~\ref{hyp:absolute-convergence-subsampling} holds true.
Suppose that $\psi \in L_2(0,b)$ and $g(z) \in \cB_b$ are related by the isometry
\eqref{eq:map-to-dB}, that is,
\begin{equation}
  \label{eq:g-expression-by-inner}
g(z) = \inner{ \xi(\cdot,\cc{z}) }{ \psi(\cdot) }_{L_2(0,b)} ,
\quad  z \in \C  .
\end{equation}
Define
\begin{equation}
\label{definicion g tilde}
\widehat{g}(z) \defeq \inner{ \xi^{ext}_a(\cdot,\cc{z}) }{ \psi(\cdot) }_{L_2(0,b)} ,
\quad  z \in \C .
\end{equation}
Then, due to Lemma~\ref{propiedades de xi extendida}\ref{lem:xi-ext-is-xi},
\begin{equation*}
\abs{ g(z) - \widehat{g}(z) } =
\abs{ \int_a^{b} \big( \xi(x,z) - \xi^{ext}_a(x,z) \big) \psi(x) \, dx }
\le h_a(z) \int_a^{b} \abs{ \psi(x) } \, dx ,
\end{equation*}
where the function $h_a$ has been defined in
Lemma~\ref{propiedades de xi extendida}\ref{lem:h-is-continuous}.
Therefore, for each $\psi \in L_2(0,b)$, the difference $\abs{g(z)-\widehat{g}(z)}$
is uniformly bounded in compact subsets of $\C$.
Below we prove that Hypothesis~\ref{hyp:absolute-convergence-subsampling}
holds true when $V$ satisfies \ref{propiedades requeridas para el
  potencial} with $b>\pi$.
As in the previous section, this is performed in two stages,
the first one deals with the particular case $V \equiv 0$ and
the second one treats the general case.

In keeping with the simplification made in the previous section, we consider only
the case $a=\pi$ and $\gamma=\pi/2$.

Using trigonometric identities and equations
\eqref{eq:reproducing-xi} and \eqref{eq:xi-without-potential}
one verifies that
\begin{equation}
\label{nucleo sin potencial lambda real}
\accentset{\circ}{k}_\pi(n^2,\cc{z})
	=\int_0^\pi\cos(nx)\cos(\sqrt{z}x)dx =\frac{(-1)^{n+1}}{n^2-z}\sqrt{z}\sin(\sqrt{z}\pi)\,.
\end{equation}
whenever $n \in \N\cup\{0\}$ and $z \in \C \setminus \{n^2\}$.
Recall that $\accentset{\circ}{k}_\pi$ denotes the reproducing kernel within $\cB_\pi$
associated with $V\equiv 0$.

\begin{proposition} \label{conv uniforme lambda real}
Hypothesis~\ref{hyp:absolute-convergence-subsampling} holds true
under the assumption $V \equiv 0$, $a=\pi$, and $\gamma=\pi/2$.
\end{proposition}
\begin{proof}
  Let $K$ be a compact subset of $\C$. As in the proof of
  Proposition~\ref{hipotesis sobremuestreo sin potencial}, assume
  without loss of generality that $n_0^2$ is the only point of
  $\spec(H_\pi(\pi/2))$ in $K$
  ($n_0\in\N$).  Due to \eqref{eq:xi-without-potential}--\eqref{norma
    funciones propias al cuadrado}, it suffices to show the uniform
  convergence of the series
  $\sum_{n\ne n_0}\abss{\accentset{\circ}{k}_\pi(n^2,\cc{z})}$
   in $K$. By \eqref{nucleo sin potencial lambda real}, one obtains
\begin{equation*}
  \sum_{n\ne n_0}\abs{\accentset{\circ}{k}_\pi(n^2,\cc{z})}
\le\abs{\sqrt{z}
\sin(\sqrt{z} \pi)} \sum_{n\ne n_0}\frac{1}{\abs{n^2 - z}}\,.
\qedhere
\end{equation*}
\end{proof}

Now we address the case of nontrivial potential $V$ satisfying
\ref{propiedades requeridas para el potencial} with $b>\pi$.
Let $\spec \left( H_\pi\left(\pi/2\right) \right)
=\{\lambda_n\}_{n=0}^\infty$ such that $\lambda_{n-1} < \lambda_n$
for all $n \in \N$.
We aim to study the difference
\begin{equation*}
\frac{k_\pi(\lambda_n,\cc{z})}{k_\pi(\lambda_n,\lambda_n)} \xi(x,\lambda_n) -
\frac{\accentset{\circ}{k}_\pi(n^2,\cc{z})}{\accentset{\circ}{k}_\pi(n^2,n^2)}\cos(n x) ,
\quad x \in [0,b] , \quad  z \in \C  ,
\end{equation*}
for any given $b>\pi$ and all $n\in\N$ large enough.

\begin{lemma}   \label{estimacion nucleo k V}
For any $V$ satisfying \ref{propiedades requeridas para el potencial}
with $b>\pi$,
there exists an $N\in \N$ such that, if $n \ge N$, then
\begin{equation*}
\abs{k_\pi(\lambda_n,\cc{z}) - \accentset{\circ}{k}_\pi(n^2,\cc{z})}
\le D_\pi\frac{e^{\abs{\im\sqrt{z}}\pi}}{n^2}
\left( 1 + \frac{1 + \abs{z}}{1+\pi\abs{z}^{1/2}} \right)\,,
\end{equation*}
for every $z \in \C$. Here $D_\pi$ is a positive real number depending on $V$.
\end{lemma}
\begin{proof}
In view of (\ref{eq:reproducing-xi}) and Definition \ref{funciones auxiliares},
\begin{align*}
k_\pi(\lambda_n,\cc{z}) &- \accentset{\circ}{k}_\pi(n^2,\cc{z})
= \int_0^\pi \Big[\cos(nx) F(x,z) + \frac{\rho(x)}{n} \sin(nx) \cos(\sqrt{z}\,x)
\\[1mm]
&+ \frac{\rho(x)}{n} \sin(nx) F(x,z) + T(x,n) \cos(\sqrt{z}\,x)
+ T(x,n) F(x,z) \Big] dx . 
\end{align*}
We proceed as in the proof of Lemma~\ref{sumandos oversampling}. The
first three terms on the right-hand side of the last equality are
estimated by Lemma~\ref{lem:integrals-1-3}. The remaining terms have
estimates obtained in the same way as the estimates
\eqref{eq:first-A7} and \eqref{eq:second-A7}.
\end{proof}

\begin{lemma}
\label{importante}
Assume that $V$ satisfies \ref{propiedades requeridas para el potencial}
with $b>\pi$. Then, the asymptotic formula
\[
\xi(x,\lambda_n) - \cos(n x) = \cO(n^{-1}) , \quad n \to \infty ,
\]
holds uniformly with respect to $x \in [0,b]$.
\end{lemma}
\begin{proof}
Using Lemma~\ref{lemma:nuevo}\ref{eq:asymp-eigenvalues} and repeating the
reasoning leading to \eqref{inverso mult normas eigen funciones}, one arrives at
\begin{equation*}
\lambda_n^{-1/2} - n^{-1} = \cO(n^{-1}) ,
\quad  n \to \infty .
\end{equation*}
This asymptotic formula and \eqref{estimado funciones propias con potencial}
yield
\begin{equation*}
\xi(x,\lambda_n) - \cos(\sqrt{\lambda_n} \, x)
	= \cO(n^{-1}) , \quad   n \to \infty .
\end{equation*}
Finally, since
\begin{equation*}
  \abs{\cos(\sqrt{\lambda_n} \, x) - \cos(n x)}=\abs{\sin(\alpha_n x)}\abs{\sqrt{\lambda_n}x-nx}\le\abs{\sqrt{\lambda_n}-n}b
\end{equation*}
for some $\alpha_n$ between $\sqrt{\lambda_n}$ and $n$, the
statement follows from Lemma~\ref{lemma:nuevo}\ref{eq:asymp-eigenvalues}.
\end{proof}

\begin{proposition}
\label{hipotesis undersampling caso general}
Let $V$ be as in \ref{propiedades requeridas para el potencial} with $b>\pi$.
Set $a=\pi$ and $\gamma=\pi/2$.
Then, Hypothesis~\ref{hyp:absolute-convergence-subsampling} holds true.
\end{proposition}

\begin{proof}
Due to Lemmas~\ref{estimacion nucleo k V} and \ref{importante}, along with
(\ref{inverso mult normas eigen funciones}), there exists $N\in \N$
and a continuous positive function $c_3:\C\to\R$ such that
\begin{equation}
\label{eq:inequality}
\abs{\frac{k_\pi(\lambda_n,\cc{z})}{k_\pi(\lambda_n,\lambda_n)}\xi(x,\lambda_n) -
\frac{\accentset{\circ}{k}_\pi(n^2,\cc{z})}{\accentset{\circ}{k}_\pi(n^2,n^2)}\cos(n\,x)}
\le
\frac{c_3(z)}{n^2} ,\quad z \in \C ,\quad x \in [0,b] .
\end{equation}
for all $n \ge N$; we note that $c_3$ may depend on $b$ and $V$.
The estimate \eqref{eq:inequality} in turn implies
\begin{equation*}
\sum_{n=0}^\infty
\abs{\frac{k_\pi(\lambda_n,\cc{z})}{k_\pi(\lambda_n,\lambda_n)} \xi(x,\lambda_n) -
\frac{\accentset{\circ}{k}_\pi(n^2,\cc{z})}{\accentset{\circ}{k}_\pi(n^2,n^2)}
\cos(n\,x)}
\le
c_4(z)
\end{equation*}
uniformly with respect to $x \in [0,b]$, where $c_4:\C\to\R$ is another continuous
positive function that may also depend on $b$ and $V$. The claimed assertion
now follows from Proposition~\ref{conv uniforme lambda real}.
\end{proof}

\begin{theorem}
\label{thm:main-subsampling}
Suppose $V$ obeys \ref{propiedades requeridas para el potencial} for $b>\pi$.
Assume that $\psi\in L_2(0,b)$ and $g(z)\in\cB_b$ are related by
\eqref{eq:g-expression-by-inner}. For 
every compact $K\subset\C$,
there exist a constant $D(b,K,V)>0$ such that
\[
\abs{g(z)-\widehat{g}(z)}
\le
D(b,K,V)\int_{\pi}^{b}\abs{\psi(x)} dx, \quad z\in K,
\]
where $\widehat{g}(z)$ is given by \eqref{definicion g tilde} with
$a=\pi$, \ie, $\widehat{g}(z)$ is given by the series \eqref{xi ext
  prod int} with $a=\pi$ and $\gamma=\pi/2$.
\end{theorem}
\section*{Acknowledgements}
The authors thank the anonymous referee
whose pertinent comments led to an improved presentation of this work.

\appendix
\section{Auxiliary results}

\begin{lemma}
\label{lemma:analyis-exercise}
Let $Y$ be a compact interval of $\R$. Suppose $\theta: \C \times Y \to [0,\infty)$ is
continuous. Then, $\Theta: \C \to [0,\infty)$ given by
$\Theta(z) \defeq \sup\{\theta(z,y) : y \in Y\}$ is continuous.
\end{lemma}

\begin{proof}
For each $z \in \C$, fix $\vartheta(z) \in Y$ such that
\begin{equation}  \label{fijo punto donde supremo}
\theta\big(z,\vartheta(z)\big)=\sup\{\theta(z,y) : y \in Y\}=\Theta(z)  \,.
\end{equation}
Take an arbitrary $z_0 \in \C$. Fix $r_0 > 0$ and let
$K \defeq \{w \in \C \,:\, \abs{z_0 - w} \le r_0 \}$.
Due to the compactness of $K \times Y$, the map
$\theta \restriction_{K \times Y}$
is uniformly continuous.
Hence, given $\epsilon>0$ there exists $\delta>0$ such that
\begin{equation}   \label{continuidad uniforme}
\abs{z-w}<\delta \mbox{ and } \abs{y-v}<\delta
\mbox{ imply } \abs{\theta(z,y)-\theta(w,v)}<\frac{\epsilon}{2} \,,
\end{equation}
for any $(z,y)\,,(w,v) \in K \times Y$.
Take $w \in K$ such that $\abs{z_0-w}<\delta$.
If $v \in Y$ satisfies $\abs{\vartheta(z_0)-v}<\delta$ then,
in view of \eqref{continuidad uniforme},
\[
\abs{\theta\big(z_0,\vartheta(z_0)\big)}-\abs{\theta(w,v)} \le
\abs{\theta\big(z_0,\vartheta(z_0)\big)-\theta(w,v)} < \frac{\epsilon}{2} \,.
\]
Due to \eqref{fijo punto donde supremo} and the fact that
$\theta$ is non negative,
$\Theta(z_0)-\Theta(w)\le \Theta(z_0) - \theta(w,v) <\epsilon$.
Now, let $v \in Y$ such that $\abs{\vartheta(w)-v}<\delta$.
According to \eqref{continuidad uniforme},
\[
\abs{\theta\big(w,\vartheta(w)\big)}-\abs{\theta(z_0,v)} \le
\abs{\theta\big(w,\vartheta(w)\big)-\theta(z_0,v)} < \frac{\epsilon}{2} \,.
\]
Hence, $\Theta(w)-\Theta(z_0) \le \Theta(w)-\theta(z_0,v)<\epsilon$.
Therefore, we have proven that $-\epsilon<\Theta(z_0)-\Theta(w)<\epsilon$
whenever $\abs{z_0-w}<\delta$.
\end{proof}

The following Lemma is the analogue of \cite[Lemma\,2.2]{MR2719774}
for Neumann-like boundary conditions.
                          
\begin{lemma}        \label{lemma:cota_eigen_functions}
Given $a>0$, suppose that $V\in L_1(0,a)$.
Then, for each $z \in \C$, the unique solution of the initial value problem
\begin{align*}
-&\xi''(x,z)+V(x)\xi(x,z)=z\xi(x,z) \,,
\qquad 0\le x\le a,
\\[2mm]
&\xi(0,z)=1,  \quad \xi'(0,z)=0, 
\end{align*}
satisfies the integral equation
\begin{equation}   \label{ecuacion integral para soluciones}
\xi(x,z) = \cos \left( \sqrt{z}x \right) + \int_0^x G(z,x,y) V(y) \xi(y,z) dy \,,
\end{equation}
where
\begin{equation*}
\label{especificamos funcion de green prob de val inic}
G(z,x,y) = \frac{1}{\sqrt{z}} \sin\left(\sqrt{z}\left(x-y\right)\right)
\end{equation*}
is the corresponding Green's function. This solution satisfies the estimate
\begin{equation}
\label{estimado funciones propias con potencial}
\abs{\xi(x,z)-\cos \left( \sqrt{z}x \right)}
\le C
	 \frac{x}{1+\abs{z}^{1/2}x}e^{\abs{\im\sqrt{z}}x}
	 \int_0^x\frac{y\abs{V(y)}}{1+\abs{z}^{1/2}y}dy
\end{equation}
for some constant $C=C(a, V)>0$. Furthermore, the derivative obeys
\begin{equation}
\qquad \qquad \xi'(x,z) = -\sqrt{z}\sin \left( \sqrt{z}x \right) + \int_0^x
\frac{\partial}{\partial x}
G(z,x,y) V(y) \xi(y,z) dy \,,
\label{ecuacion integral para derivada de soluciones}
\end{equation}
and satisfies the estimate
\begin{equation}
\abs{\xi'(x,z)+\sqrt{z}\sin \left( \sqrt{z}x \right)} \le
C e^{\abs{\im\sqrt{z}}x}\int_0^x\abs{V(y)}dy
\label{estimado derivada funciones propias con potencial}.
\end{equation}
\end{lemma}
\begin{proof}
Define
\[
\xi_0(x,z) \defeq \cos \left( \sqrt{z}x \right),
\quad
\xi_{n+1}(x,z) \defeq \int_{0}^{x}G(z,x,y)V(y)\xi_n(y,z)dy
\label{sumandos que componen a phi},
\quad n \in \N.
\]
Since $\abs{\cos \left( \sqrt{z}x \right)} \le \exp(\abs{\im\sqrt{z}}x)$
and
\begin{equation*}
\abs{G(z,x,y)}
\le C_0 \, \frac{x}{1+\abss{z}^{1/2}x}e^{\abs{\im\sqrt{z}}(x-y)},
\quad 0 \le y \le x,
\end{equation*}
 for some constant $C_0>0$ (\cf \cite[Lemma\,A.1]{MR2719774}), one has
\[
\abs{\xi_{1}(x,z)}
\le
C_0\norm{V}_{L_1} \frac{x}{1+\abss{z}^{1/2}x}e^{\abs{\im\sqrt{z}}x}.
\]
An induction argument then shows
\begin{equation}
\label{cota para sumando generico}
\abs{\xi_{n+1}(x,z)}
\le
	\frac{\norm{V}_{L_1}C_0^{n+1}}{(n+1)!}
	\frac{x}{1+\abss{z}^{1/2}x}e^{\abs{\im\sqrt{z}}x}
	\left(\int_{0}^{x}\frac{y\abs{V(y)}}{1+\abs{z}^{1/2}y}dy\right)^{n}
\end{equation}
for all $n\in\N$. It follows that
\[
\xi(x,z) \defeq \sum_{n=0}^\infty \xi_n(x,z)
\]
converges uniformly with respect to $x\in[0,a]$ for all $z\in\C$ and satisfies
\eqref{ecuacion integral para soluciones}.
The estimate \eqref{estimado funciones propias con potencial} readily follows from
\eqref{cota para sumando generico} after noticing that
\[
\int_0^x\frac{y\abs{V(y)}}{1+\abs{z}^{1/2}y}dy
	\le a\norm{V}_{L_1}.
\]
The assertions \eqref{ecuacion integral para derivada de soluciones}
and \eqref{estimado derivada funciones propias con potencial} are proved
by similar arguments so we omit the details.
\end{proof}
The next results refer to the functions $\rho$, $T$, and $F$ introduced in
Definition~\ref{funciones auxiliares}, as well as the reproducing kernel
$k_b(z,w)$ from \eqref{eq:reproducing-xi} and the particular case
$\accentset{\circ}{k}_b(z,w)$ when $V\equiv 0$.

\begin{lemma}
\label{lemma:nuevo}
Assume that $V$ satisfies \ref{propiedades requeridas para el
  potencial} with $b=\pi$.
Let $H_{\pi}(\pi/2)$ be the selfadjoint operator defined in accordance with
\eqref{selfadjoint extensions}. Enumerate $\spec(H_{\pi}(\pi/2))$ in increasing order
and denote $\spec(H_{\pi}(\pi/2))=\{\lambda_n\}_{n=0}^\infty$. Then, the following
assertions hold true.
\begin{enumerate}[label={(\roman*)}]
\item $\sqrt{\lambda_n}=n+\cO(n^{-1})$ as $n\to\infty$, \label{eq:asymp-eigenvalues}
\item $T(x,n)=\cO(n^{-2})$ as $n\to\infty$, uniformly with respect to
	$x\in[0,\pi]$,\label{eq:asymp-T}
\item $k_\pi(\lambda_n,\lambda_n)=\accentset{\circ}{k}_\pi(n^2,n^2)+\cO(n^{-2})$
	as $n\to\infty$. \label{eq:asymp-k}
\end{enumerate}
\end{lemma}
\begin{proof}
Items \ref{eq:asymp-eigenvalues} and \ref{eq:asymp-T} are shown in
\cite[Sec.1.2.2]{MR1136037}. We note that the asymptotic formulae in
\cite{MR1136037} are obtained assuming that $V'$ is bounded in
$[0,\pi]$. However, one can see that it suffices to require $V'\in
L_1(0,\pi)$.

We turn to the proof of \ref{eq:asymp-k}. Let us recall that
\begin{equation*}
k_\pi(\lambda_n,\lambda_n)=\inner{\xi(\cdot,\lambda_n)}{\xi(\cdot,\lambda_n)}_{L_2(0,\pi)}
	=\int_0^\pi\abs{\xi(x,\lambda_n)}^2dx=\int_0^\pi\xi^2(x,\lambda_n)dx,
\end{equation*}
while
\begin{equation*}
  \accentset{\circ}{k}_\pi(n^2,n^2)=\int_0^\pi (\cos nx)^2dx.
\end{equation*}
A straightforward computation shows that
\begin{equation*}
\sup_{0\le x\le\pi} \abs{\rho(x)} \le \norm{V}_{L_1(0,\pi)}, \quad 
\sup_{0\le x\le\pi} \abs{\rho'(x)} \le \norm{V}_{L_1(0,\pi)}.
\end{equation*}
Together with \eqref{eq:T-def}
and
\ref{eq:asymp-T}, these inequalities imply
\begin{equation}    \label{eq:square-eigenfunctions-simplified}
\xi^2(x,\lambda_n)=(\cos nx)^2+\frac{\rho(x)}{n}\sin(2nx)+\cO(n^{-2}), \qquad n\to\infty,
\end{equation}
uniformly with respect to $x\in[0,\pi]$. Using integration by parts
along with the fact that $\rho(0)=\rho(\pi)=0$, one obtains
\begin{equation}    \label{eq:bound-remaining-integral}
\abs{\int_0^\pi \rho(x)\sin(2nx)dx} \le \frac{1}{2n}\int_0^\pi \abs{\rho'(x)\cos(2nx)}dx.
\end{equation}
Assertion \ref{eq:asymp-k} follows from
\eqref{eq:square-eigenfunctions-simplified} and \eqref{eq:bound-remaining-integral}.
\end{proof}
\begin{lemma}
\label{lem:integrals-1-3}
Assume that $V$ satisfies \ref{propiedades requeridas para el
  potencial} with $b=\pi$.
Consider an arbitrary $a \in (0,\pi]$.
Then, for all $z \in \C$ and $n \in \N$, the following inequalities hold true:
\begin{gather}
\abs{\int_0^\pi F(x,z) \cos(n\,x) \,dx}
\le C_1 \frac{e^{\pi \abs{\im\sqrt{z}}}}{n^2}
	\left(1 + \frac{1 + \abs{z}}{1 + \pi\abs{z}^{1/2}}\right) ,
	\label{eq:int1}
\\[2mm]
\abs{\int_0^a \rho(x) \cos(\sqrt{z}\,x) \sin(nx) dx}
\le C_2 \frac{e^{\pi \abs{\im\sqrt{z}}}}{n}
	\left( 1 + \frac{\abs{z}}{1+\pi\abs{z}^{1/2}}\right) ,
	\label{eq:int2}
\\[2mm]
\abs{\int_0^a \rho(x) F(x,z) \sin(nx) dx}
\le C_3 \frac{e^{\pi \abs{\im\sqrt{z}}}}{n}
	\left(1 + \frac{1}{1+\abs{z}^{1/2}\pi}\right) .
	\label{eq:int3}
\end{gather}
Here, $C_1>0$ depends on $V$ while $C_2>0$ and $C_3>0$ may, in addition,
depend on $a$.
\end{lemma}
\begin{proof}
Integrating by parts one obtains,
\begin{equation*}
\abs{\int_0^\pi F(x,z) \cos(nx) dx}
\le \frac{1}{n^2} \left( 2 \sup_{x\in[0,\pi]} \abs{F'(x,z)}
	+ \pi \sup_{x\in[0,\pi]} \abs{F''(x,z)} \right) .
\end{equation*}
On one hand, due to \eqref{estimado derivada funciones propias con potencial},
\begin{equation*}
\sup_{x\in[0,\pi]} \abs{F'(x,z)} \le C_V
\exp( \abs{ \im \sqrt{z} } \pi ) .
\end{equation*}
On the other hand, since $F''(x,z)=V(x)\xi(x,z)-zF(x,z)$, it follows
from \eqref{estimado funciones propias con potencial} that
\begin{equation*}
\abs{F''(x,z)}
\le e^{\pi \abs{\im\sqrt{z}}}
	\left( \frac{C_V\,x}{1+\abs{z}^{1/2}x}\left( \norm{V}_{L_1}
	+\abs{z}\right) + \norm{V}_{L_1}\right) .
\end{equation*}
This implies \eqref{eq:int1}.

The proof of \eqref{eq:int2} repeats the argumentation above: integrate by parts and
observe that
\begin{gather*}
\sup_{x\in[0,a]} \abs{\rho(x) \cos(\sqrt{z}\,x)}
	 \le \norm{V}_{L_1} e^{\pi \abs{\im\sqrt{z}}} ,
\\[1mm]
\sup_{x\in[0,a]} \abs{\frac{d}{dx}\rho(x) \cos(\sqrt{z}\,x)}
	\le \norm{V}_{L_1} e^{\pi \abs{\im\sqrt{z}}}
		\left( \frac{\abs{z}\,C\,\pi}{1+\abs{z}^{1/2}\pi} + 1 \right) .
\end{gather*}

The proof of \eqref{eq:int3} follows a similar reasoning.
\end{proof}

\begin{lemma}
\label{valor importante del producto interior}
Set $a \in (0,\pi)$ and consider $\mathcal{R}_{a\pi}$ given by
\eqref{def de funcion R}. Then, for any $n \in \N\cup\{0\}$ and
$z \in \C \setminus \{ n^2 \}$,
\begin{align*}
\inner{\cos(\sqrt{\cc{z}}\,\cdot)}{\cR_{a\pi}(\cdot)\cos(n\,\cdot)}
&= \frac{1}{2 (\pi-a)}
	\Bigg(\frac{\cos\big( (\sqrt{z}+n)a \big)-(-1)^n \cos(\sqrt{z} \, \pi)}
			   {(\sqrt{z}+n)^2}\\[1mm]
&+ \frac{\cos \big( (\sqrt{z}-n)a \big)-(-1)^n \cos(\sqrt{z} \, \pi)}
	   {(\sqrt{z}-n)^2} \Bigg) .
\end{align*}
\end{lemma}
\begin{proof}
On one hand, the identity
\[
\cos(\sqrt{z} x) \cos(n x)
= 2^{-1}\big(\cos( (\sqrt{z}+n) x ) + \cos((\sqrt{z}-n)x)\big)
\]
leads to
\begin{equation*}
\int_0^{a} \cos(\sqrt{z} x) \cos(n x) \, dx
	= \frac{1}{2} \left( \frac{ \sin \big( (\sqrt{z}+n)a \big)}{\sqrt{z}+n}
		+ \frac{ \sin \big( (\sqrt{z}-n)a \big)}{\sqrt{z}-n} \right) .
\end{equation*}
On the other hand,
\begin{multline*}
\int_{a}^\pi \cos(\sqrt{z} x) \cos(n x)
\left( \frac{\pi-x}{\pi-a} \right) dx
\\[2mm]
= - \frac{1}{2(\pi-a)} \left(\int_{a}^\pi x \cos \big( (\sqrt{z}+n) x \big) dx
	+ \int_{a}^\pi x \cos \big( (\sqrt{z}-n) x \big) dx \right)
\\[2mm]
+\frac{\pi}{2(\pi-a)} \left( \frac{\sin \big( (\sqrt{z}+n) x \big)}{\sqrt{z}+n}
	+ \frac{\sin\big((\sqrt{z}-n) x \big)}{\sqrt{z}-n} \right)\Bigg\rvert_{x=a}^{x=\pi} .
\end{multline*}
Another integration by parts yields
\begin{multline*}
\int_{a}^\pi x \cos ( (\sqrt{z} \pm n)  x ) \, dx
= \frac{(-1)^n \cos(\sqrt{z} \pi)
	- \cos \big( (\sqrt{z} \pm n) a \big)}{(\sqrt{z} \pm n)^2}
\\[2mm]
+ \frac{(-1)^n \sin(\sqrt{z} \pi)
	- a \sin \big( (\sqrt{z} \pm n) a \big)}{\sqrt{z} \pm n}  .
\end{multline*}
This completes the proof.
\end{proof}

\begin{lemma}
\label{int 1 oversam}
Set $a \in (0,\pi)$ and consider $\mathcal{R}_{a\pi}$ given by
\eqref{def de funcion R}. Then, for every $z\in\C$ and $n\in\N$,
\begin{equation*}
\abs{\int_0^\pi F(x,z) \cR_{a\pi}(x) \cos(nx) dx}
	\le C\,\frac{e^{\pi \abs{\im\sqrt{z}}}}{n^2}
		   \left(1 + \frac{1 + \abs{z}}{1 + \pi\abs{z}^{1/2}} \right) ,
\end{equation*}
where $C > 0$ may depend on $V$.
\end{lemma}
\begin{proof}
Integration by parts yields
\begin{gather}
\label{int F(x,z) cos(nx) zero beta}
\int_0^a F(x,z) \cos(nx)\,dx
	= \frac{1}{n}\Big( F(a,z)\sin(na) - \int_0^a F'(x,z)\sin(nx)\, dx \Big) ,
\\[2mm]
\label{int F(x,z) cos(nx) beta pi}
\int_a^\pi\!\! F(x,z) \cos(nx)\, dx
	= - \frac{1}{n}\Big( F(a,z)\sin(na)
		+\!\int_a^\pi\!\! F'(x,z)\sin(nx)\, dx \Big) ,
\end{gather}
and
\begin{multline}
\int_a^\pi x F(x,z) \cos(nx) dx
	= - \frac{1}{n}\Big( a F(a,z)\sin(na)
\\[2mm]
	 + \int_a^\pi x F'(x,z) \sin(nx)\,dx
     + \int_a^\pi F(x,z) \sin(nx)\,dx\Big)
\label{int x F(x,z) cos(nx)}  .
\end{multline}
Now, \eqref{int F(x,z) cos(nx) beta pi} and \eqref{int x F(x,z) cos(nx)} imply
\begin{multline}
\label{integral larga}
\int_a^\pi F(x,z) \left(\frac{\pi-x}{\pi-a}\right)\cos(nx) \,dx
\\[2mm]
	= -\,\frac{\pi}{(\pi-a)n}\Big(F(a,z) \sin(na)
      +  \int_a^\pi F'(x,z)\sin(nx) dx\Big)\qquad
\\[2mm]
      + \frac{1}{(\pi-a)n}\Big(a F(a,z) \sin(na)
      + \int_a^\pi xF'(x,z)\sin(nx) dx
\\[2mm]
      + \int_a^\pi F(x,z)\sin(nx) dx\Big) .
\end{multline}
Then, \eqref{int F(x,z) cos(nx) zero beta} and \eqref{integral larga} yield
\begin{multline*}
\int_0^\pi F(x,z) \cR(x) \cos(nx)\, dx
\\[2mm]
	= - \frac{1}{n}\int_0^a F'(x,z)\sin(nx)\, dx
      - \frac{\pi}{(\pi-a)n}\int_a^\pi F'(x,z)\sin(nx)\, dx
\\[2mm]
      + \frac{1}{(\pi-a)n}\Big(\int_a^\pi xF'(x,z)\sin(nx)\, dx
      + \int_a^\pi F(x,z)\sin(nx) dx\Big) .
\end{multline*}
The claimed assertion now follows by an argument similar to the proof of
Lemma~\ref{lem:integrals-1-3}.
\end{proof}

\begin{lemma}
\label{int 2 oversam}
Let $V$ as in \ref{propiedades requeridas para el potencial} with $b=\pi$. Fix $a \in (0,\pi)$.
Then,
\[
\abs{\int_a^\pi \cos(\sqrt{z}\,x) \left(\frac{\pi-x}{\pi-a}\right)
\rho(x) \sin(nx)\, dx}
	\le \frac{C}{n} e^{\pi \abs{\im\sqrt{z}}}
		\left(1
		+ \frac{\abs{z}}{1 + \pi\abs{z}^{1/2}} \right) ,
            \]
            and

\[
\abs{\int_a^\pi F(x,z) \left(\frac{\pi-x}{\pi-a}\right) \rho(x)\sin(nx)\, dx}
\le \frac{C}{n}\,e^{\pi \abs{\im\sqrt{z}}}
	\left(1 + \frac{1}{1+\pi\abs{z}^{1/2}}\right) ,
\]
for arbitrary $z \in \C$ and $n \in \N$.
\end{lemma}
\begin{proof}
We prove the first inequality. The second one is proved analogously.
Arguing as in the beginning of the proof of Lemma~\ref{lem:integrals-1-3}, one obtains
\begin{equation*}
\abs{\int_a^\pi \cos(\sqrt{z}\,x) \left(\frac{\pi-x}{\pi-a}\right) \rho(x) \sin(nx)\,dx}
\le
\frac{1}{n}\left(2 M_1(z) + \pi M_2(z)\right) ,
\end{equation*}
where
\begin{equation*}
M_1(z) \defeq \sup\left\{ \abs{\cos(\sqrt{z}\,x)
			\frac{\pi-x}{\pi-a} \rho(x)} : x\in [a,\pi]\right\} ,
\end{equation*}
and
\begin{equation*}
M_2(z) \defeq \sup\left\{\abs{\frac{d}{dx}
			\Big( \cos(\sqrt{z}\,x)
			\frac{\pi-x}{\pi-a} \rho(x) \Big)} : x\in[a,\pi]\right\} .\qedhere
\end{equation*}
\end{proof}

\end{document}